\documentclass[conference]{IEEEtran}
\usepackage{amsmath,amssymb,amsfonts}
\usepackage{bm}
\usepackage{graphicx}
\usepackage{amsthm}
\usepackage{algorithm}
\usepackage{algpseudocode}
\usepackage{cite}
\graphicspath{{./}}
\newtheorem{proposition}{Proposition}
\title{Factor-Wise Parameter Learning in Score-Based VAMP}

\author{\IEEEauthorblockN{Siqi Na and Tadashi Wadayama}
  \IEEEauthorblockA{Department of Computer Science,
  Nagoya Institute of Technology, Nagoya, Japan\\
  E-mail: \{nasiqi, wadayama\}@nitech.ac.jp}}

\begin{document}
\maketitle

\begin{abstract}
Score-based vector approximate message passing (SC-VAMP)
builds state-evolution-predictable inference from modular soft-input
soft-output (SISO) factors, but every factor must be handed
correct parameters: a wrong sparsity rate, signal variance, or noise level
distorts the tilted distributions and Onsager terms, raises the fixed-point
mean-squared error (MSE), and
invalidates matched state-evolution (SE) prediction. We let each factor learn
its own parameters instead. An expectation--maximization (EM) M-step taken
under the module's tilted distribution is attached to every parameterized
factor and costs $O(N)$ per module, because that distribution and its moments
are precisely what the SISO rule already
computes on the way to the score; the score-based interface and the Onsager
correction are left untouched. Under standard large-system SE assumptions and
identifiability conditions, the true parameters and the matched SE variance form
a fixed point of the recursion, at which the estimate attains the replica
minimum MSE (MMSE); prior factors and
likelihood/linear-MMSE (LMMSE) factors are treated
separately, the latter
under the Gaussian cavity induced by the VAMP transformed-error model. On linear
and one-bit Bernoulli--Gaussian compressed sensing, the updates recover
near-oracle accuracy from strongly mismatched initializations.
\end{abstract}
\begin{IEEEkeywords}
Vector approximate message passing, score function,
expectation--maximization, parameter estimation, one-bit quantization
\end{IEEEkeywords}

\section{Introduction}

High-dimensional Bayesian inverse problems require both an inference method
and a statistical model for the unknown signal and the observation process.
In practice, parameters of this model, such as sparsity rates,
component variances, and noise levels, are seldom known
beforehand and must be estimated
jointly with the signal. This uncertainty is consequential for iterative
Bayesian inference: a parameter mismatch changes the local posterior
approximation and propagates the resulting error through subsequent messages.
The problem arises in compressed sensing (CS), imaging,
communications, and
quantized sensing, where the observation model may be linear or nonlinear and
the sensing operator may be far from an ideal i.i.d. random matrix.

Approximate message passing (AMP) provides an efficient approach to such
high-dimensional estimation problems and admits a scalar state evolution (SE)
for large i.i.d. sensing matrices \cite{Donoho2009AMP}. Generalized AMP (GAMP)
incorporates separable nonlinear observation channels
\cite{Rangan2011GAMP}, whereas vector AMP (VAMP) retains SE-predictable
inference for the broader class of right-rotationally invariant (RRI) matrices
\cite{Rangan2019VAMP}, where standard AMP/GAMP can be fragile
\cite{Rangan2019AMPConv}. Score-based VAMP (SC-VAMP) is a modular
reformulation of VAMP in which a prior, likelihood, or coupling factor is
represented by a common soft-input soft-output (SISO) module
\cite{Wadayama2026ScoreVAMP}. Given an incoming Gaussian cavity, a module forms
the tilted distribution obtained by combining that cavity with its factor. Its
score and moments determine the posterior estimate, the Onsager coefficient,
and the outgoing extrinsic Gaussian message. Consequently, heterogeneous
factors can be composed without re-deriving the surrounding message-passing
protocol.

A module's score may itself be learned from data. For the signal prior in
compressed sensing this is rarely an option, since $\bm x$ is the estimand and
samples for learning its score are unavailable; the prior must be given as an
explicit parametric family.
This modular inference interface does not by itself resolve model uncertainty.
Every tilted distribution still depends on the parameters of its associated
factor, and a mismatched parameter affects both the module estimate and the
extrinsic variance communicated to the remainder of the graph. For example, in
Bernoulli--Gaussian (BG) compressed sensing, errors in the sparsity rate,
active-component variance, or measurement-noise variance alter the denoising
and observation modules simultaneously and can move the algorithm away from
its matched-SE operating point \cite{Takahashi2022Mismatch}.

Parameter learning has been incorporated into several message-passing
architectures. EM-GAMP estimates prior and noise parameters in generalized
linear models \cite{Vila2013EMGMAMP}; EM-VAMP and adaptive VAMP provide
learning rules and consistency results for linear RRI systems
\cite{Fletcher2016EMVAMP,Fletcher2018Adaptive}; and EM-GVAMP treats nonlinear
output channels within generalized VAMP \cite{Metzler2018EMGVAMP}. Risk-based
methods instead tune denoisers through criteria such as Stein's unbiased risk
estimate \cite{Mousavi2018Parameterless}. These methods establish the value of
joint inference and parameter estimation within important AMP/VAMP
architectures. The question considered here is how to express adaptation at
the same factor boundary as the SC-VAMP SISO rule, so that the update depends
only on a module's factor and the tilted moments already available within
that module.
Such a rule would let a parameterized factor be moved into a
different factor graph without re-deriving either its M-step or the surrounding
protocol, and would hold the algorithm at its matched-SE operating point
without any hand-tuned input.

This paper develops factor-wise parameter learning for SC-VAMP. First, we
attach an expectation--maximization update to each parameterized factor. The
update is computed locally from the tilted moments already produced by the
SISO rule, adding only $O(N)$ work per module without changing the score-based
message interface or its Onsager correction; we refer to the resulting
algorithm as adaptive SC-VAMP. Second, under standard
large-system SE and identifiability assumptions, we show that the true factor
parameters together with the matched SE variances form a population fixed
point. The analysis treats prior factors and likelihood/LMMSE factors
separately, with the latter justified under the Gaussian cavity induced by the
VAMP transformed-error model. Third, experiments on linear and one-bit BG
compressed sensing demonstrate recovery close to the oracle-parameter baseline
from strongly mismatched initializations, while also exposing the scale
identifiability limitation of one-bit observations.

\section{Adaptive SC-VAMP}
\label{sec:method}
SC-VAMP estimates $\bm x\sim p_X$ from
\begin{equation}
\bm y=\phi(\bm A\bm x+\bm w),\qquad
\bm w\sim\mathcal N(\bm 0,\sigma_w^2\bm I),
\label{eq:model}
\end{equation}
with RRI $\bm A\in\mathbb R^{M\times N}$, by composing modules through a
uniform mean--variance interface.
Throughout, $\bm A$ and the form of $\phi$ are known; unknown
are $\bm x$, the parameters of its prior $p_X$, and the noise variance
$\sigma_w^2$, all of which are estimated jointly.

\subsection{Preliminary: SISO module}
We recall the SC-VAMP SISO rule \cite{Wadayama2026ScoreVAMP}. A module owns one
factor $\phi(\bm u)$ over $\bm u\in\mathbb R^N$. Given an incoming Gaussian
message $(\bm r,v)$, interpreted as $\bm r=\bm u+\sqrt v\,\bm n$ with
$\bm n\sim\mathcal N(\bm 0,\bm I)$, it forms
$\tilde p(\bm u)\propto\phi(\bm u)\,\mathcal N(\bm u;\bm r,v\bm I)$ with
log-partition
$\log Z(\bm r)=\log\int\phi(\bm u)\mathcal N(\bm u;\bm r,v\bm I)\,d\bm u$
and score $\bm s=\nabla_{\bm r}\log Z$, and returns
\begin{equation}
\begin{aligned}
\bm r_{\mathrm{post}}&=\bm r+v\bm s,
&\alpha&=1-\tfrac{v}{N}J,\ \ J=\mathbb E\|\bm s\|^2,\\
v_{\mathrm{out}}&=\tfrac{\alpha}{1-\alpha}v,
&\bm r_{\mathrm{out}}&=\tfrac{1}{1-\alpha}(\bm r_{\mathrm{post}}-\alpha\bm r).
\end{aligned}
\label{eq:siso}
\end{equation}
here $\bm r_{\mathrm{post}}=\mathbb E_{\bm u\sim\tilde p}[\bm u]$ is the Tweedie posterior
mean, $J=\mathbb E\|\bm s\|^2$ the Fisher information removed by the Onsager
coefficient $\alpha$, and $(\bm r_{\mathrm{out}},v_{\mathrm{out}})$ the extrinsic
message
\cite{Wadayama2026ScoreVAMP}. The adaptive step below keeps this interface unchanged
and only updates factor parameters.

\subsection{Self-tuning iteration}
Let parameterized module $m$ own a factor
$\phi_m(\bm u_m,\bm y_m;\bm\theta_m)$, where $\bm y_m$ is absent for a prior and
denotes likelihood data otherwise. Given Gaussian cavity message
$(\bm r_m,v_m)$, it forms
\begin{equation}
\begin{aligned}
\tilde p_m(\bm u_m\mid\bm y_m,\bm r_m;\bm\theta_m)
&\propto
\phi_m(\bm u_m,\bm y_m;\bm\theta_m)\\
&\quad{}\times \mathcal N(\bm u_m;\bm r_m,v_m\bm I).
\end{aligned}
\label{eq:module-tilt}
\end{equation}
Beyond the SISO update, the module re-estimates its own factor by a local
EM M-step taken under the tilted distribution,
\begin{equation}
\hat{\bm\theta}_m=\arg\max_{\bm\theta_m}\
\mathbb E_{\bm u_m\sim\tilde p_m}\big[\log\phi_m(\bm u_m,\bm y_m;\bm\theta_m)\big],
\label{eq:mstep}
\end{equation}
the expectation over the latent $\bm u_m\sim\tilde p_m$ alone ($\bm y_m,\bm r_m$
are given data). As $\log\phi_m$ separates over coordinates, this expectation is
the per-coordinate sum the algorithm computes (e.g.\ $\hat\rho=\tfrac1N\sum_i\gamma_i$
below); in the large-system limit it concentrates on the factor's population
log-likelihood, maximized at the true parameter (Proposition~\ref{prop:fp}). The
extra EM cost is $O(N)$ per module because the required tilted moments are already
computed by the SISO update. The global parameter vector is damped as
$\bm\theta\gets(1-\beta)\bm\theta+\beta\hat{\bm\theta}$, and one sweep applies
\eqref{eq:siso} and \eqref{eq:mstep} at every module, summarized in
Alg.~\ref{alg:adaptive}.

\begin{algorithm}[t]
\caption{Adaptive SC-VAMP}
\label{alg:adaptive}
\begin{algorithmic}[1]
\Require modules $\{\phi_m\}$; initial $\bm\theta$; damping $\beta$; iterations $T$
\For{$t=1,\dots,T$}
  \For{each module $m$ in schedule order}
    \State update its extrinsic message by the SISO map \eqref{eq:siso}
    \State $\hat{\bm\theta}_m\gets$ M-step \eqref{eq:mstep} of factor $\phi_m$
  \EndFor
  \State $\bm\theta\gets(1-\beta)\bm\theta+\beta\,\hat{\bm\theta}$
\EndFor
\State \Return posterior mean $\hat{\bm x}$ of the prior module
\end{algorithmic}
\end{algorithm}

\subsection{Instantiations}
For the BG prior
$p_X(x_i)=(1-\rho)\delta(x_i)+\rho\mathcal N(x_i;0,\sigma_x^2)$, the prior
module is shared by the linear and one-bit experiments. With incoming
$(\bm r_A,v_A)$ and $r_i\equiv r_{A,i}$,
\begin{equation}
\begin{aligned}
D_i&=(1{-}\rho)\mathcal N(r_i;0,v_A)+\rho\mathcal N(r_i;0,\sigma_x^2{+}v_A),\\
\gamma_i&=\rho\mathcal N(r_i;0,\sigma_x^2{+}v_A)/D_i,\\
\mu_i&=\frac{\sigma_x^2}{\sigma_x^2+v_A}r_i,\qquad
\nu=\frac{\sigma_x^2v_A}{\sigma_x^2+v_A}.
\end{aligned}
\end{equation}
Then $\hat x_i=\gamma_i\mu_i$, and the EM prior update is
\begin{equation}
\hat\rho=\frac1N\sum\nolimits_i\gamma_i,\qquad
\hat\sigma_x^2=\frac{\sum_i\gamma_i(\mu_i^2+\nu)}{\sum_i\gamma_i}.
\end{equation}
For linear CS, the LMMSE observation module gives the noise update
\begin{equation}
\hat\sigma_w^2=\frac1M\|\bm y-\bm A\bm x_{\mathrm p}\|^2+c\,\sigma_w^2,\qquad
c=\frac{v_B}{v_B+\sigma_w^2},
\label{eq:noise-mstep}
\end{equation}
where $\bm x_{\mathrm p}$ is the LMMSE posterior mean and
$v_B$ the variance of the incoming message from the prior module.
For one-bit CS, we use the probit factor
$p(y_i\mid z_i;\sigma_w^2)=\Phi(y_i z_i/\sigma_w)$ in the same SISO interface;
since only $\sigma_x^2/\sigma_w^2$ is identifiable, we fix $\sigma_w^2{=}1$ and
tune $(\rho,\sigma_x^2)$.

\subsection{Fixed-point optimality}
One sweep induces the population map
$(\bm\theta,\bm v)\mapsto\bigl(\mathcal M(\bm\theta,\bm v),\,F(\bm v;\bm\theta)\bigr)$,
where $F$ is the matched SE variance recursion and $\mathcal M$ is the large-system
limit of the per-module M-steps \eqref{eq:mstep} (the coordinate sums
concentrating on channel expectations). Let
$\bm v^\star=F(\bm v^\star;\bm\theta^0)$. We use the following standard AMP/VAMP
large-system assumptions in the RRI setting: (A1) the scalar-equivalent model
\cite[Thm.~1]{Rangan2019VAMP}, in which each module's input is its true variable
plus i.i.d.\ Gaussian noise of the SE-tracked variance; (A2)
empirical/population averaging of per-component statistics; (A3) exact
scores, so SC-VAMP reduces to VAMP, rigorous here for the linear model
\cite[Sec.~III-D]{Wadayama2026ScoreVAMP}; and (A4) a unique replica fixed point,
under which the matched SE fixed point is Bayes-optimal
\cite[Thm.~3]{Rangan2019VAMP}. We also assume each parameterized factor is a
normalized density, correctly specified and identifiable; for likelihood
factors, that the Gaussian-cavity auxiliary law has a finite tilted normalizer
and identifies $\bm\theta_m^0$ through the resulting weighted conditional KL,
the local compatibility condition the EM proof needs; and, for mixed priors
such as BG, that the prior case is read with respect to the natural latent
support variable.

\begin{figure*}[t]
\centering
\includegraphics[width=\textwidth]{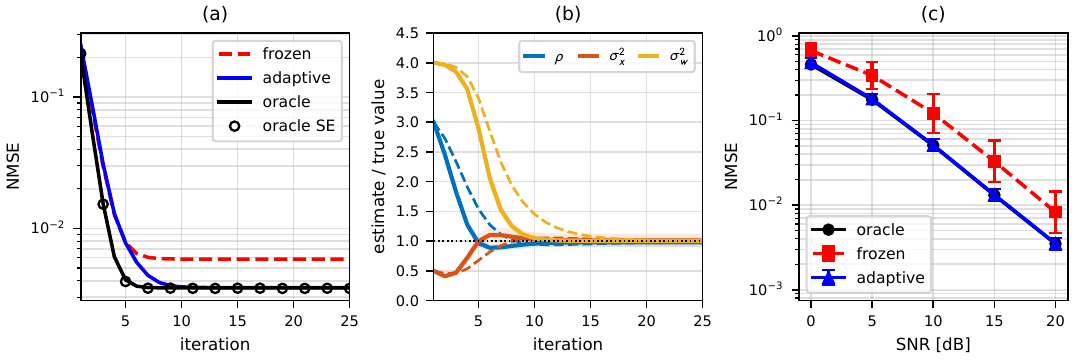}
\caption{Linear BG CS: (a)~NMSE vs.\ iteration from a mismatched start;
(b)~parameter estimates vs.\ iteration under the same start;
(c)~NMSE vs.\ SNR with per-trial randomized initialization.
Oracle uses the true parameters; frozen keeps the wrong initial ones.}
\label{fig:lin}
\end{figure*}

\begin{figure*}[t]
\centering
\includegraphics[width=\textwidth]{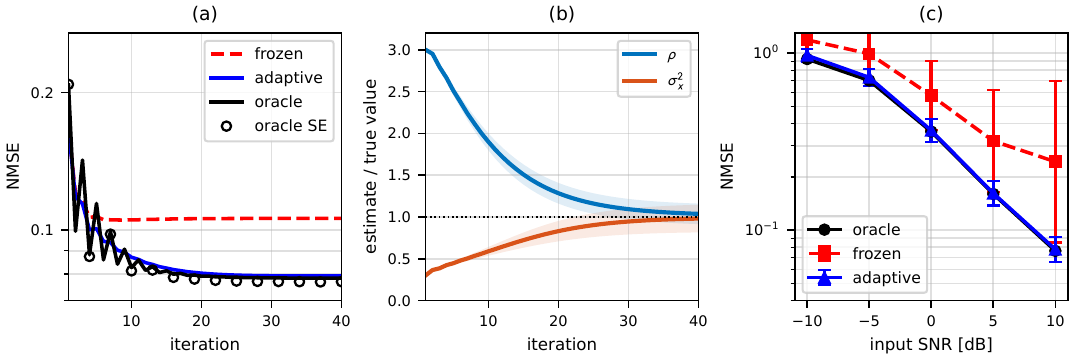}
\caption{One-bit CS: (a)~NMSE vs.\ iteration under a mismatched start;
(b)~parameter estimates vs.\ iteration; (c)~NMSE vs.\ input SNR with
per-trial randomized initialization.}
\label{fig:1bit}
\end{figure*}

\begin{proposition}\label{prop:fp}
Under the above assumptions, $(\bm\theta^0,\bm v^\star)$ is a fixed point of the
population adaptive SC-VAMP recursion. At this matched fixed point, the
posterior estimate $\hat{\bm x}$ attains the replica (Bayes) MMSE.
\end{proposition}

\begin{proof}
The variance component is fixed by definition:
\begin{equation}
\bm v^\star=F(\bm v^\star;\bm\theta^0).
\label{eq:proof-se-fixed}
\end{equation}
We prove that the population M-step also leaves each true parameter unchanged.

Fix a module $m$. If it is a prior module, write
$\phi_m(\bm u_m;\bm\theta_m)$ for the prior density. At the matched SE fixed
point,
\begin{equation}
\bm R_m=\bm U_m+\sqrt{v_m^\star}\bm Z_m,\qquad
\bm U_m\sim\phi_m(\cdot;\bm\theta_m^0),
\label{eq:proof-prior-channel}
\end{equation}
where $v_m^\star$ is the module-$m$ component of $\bm v^\star$. Write
$\phi_{\bm\theta_m}(\bm u_m)=\phi_m(\bm u_m;\bm\theta_m)$ and
$\phi_{\bm\theta_m^0}(\bm u_m)=\phi_m(\bm u_m;\bm\theta_m^0)$, and set
$g_m^\star(\bm u_m\mid\bm r_m)=
\mathcal N(\bm u_m;\bm r_m,v_m^\star\bm I)$. By Gaussian symmetry,
$g_m^\star(\bm u_m\mid\bm r_m)=
\mathcal N(\bm r_m;\bm u_m,v_m^\star\bm I)$. Hence the marginal density of
$\bm R_m$ under \eqref{eq:proof-prior-channel} is
\begin{equation}
p_{R,m}^0(\bm r_m)
=
\int
\phi_{\bm\theta_m^0}(\bm u'_m)
g_m^\star(\bm u'_m\mid\bm r_m)
\,d\bm u'_m .
\label{eq:proof-prior-r-marginal}
\end{equation}
The matched tilted density is obtained by multiplying the prior factor and the
Gaussian message, then normalizing in $\bm u_m$:
\begin{equation}
\tilde p_m^0(\bm u_m\mid\bm r_m)
=
\frac{
\phi_{\bm\theta_m^0}(\bm u_m)g_m^\star(\bm u_m\mid\bm r_m)}
{p_{R,m}^0(\bm r_m)}.
\label{eq:proof-prior-tilt}
\end{equation}
Thus, in this matched scalar channel, the tilted density coincides with the
conditional density of $\bm U_m$ given $\bm R_m=\bm r_m$.
Substituting \eqref{eq:proof-prior-tilt} into the population M-step objective
and integrating out $\bm r_m$ gives
\begin{align}
Q_m(\bm\theta_m)
&=
\int p_{R,m}^0(\bm r_m)
\int \tilde p_m^0(\bm u_m\mid\bm r_m)
\notag\\[-0.3em]
&\hspace{5em}{}
\log\phi_{\bm\theta_m}(\bm u_m)
\,d\bm u_m\,d\bm r_m\notag\\
&=
\iint
\frac{
p_{R,m}^0(\bm r_m)
\phi_{\bm\theta_m^0}(\bm u_m)g_m^\star(\bm u_m\mid\bm r_m)}
{p_{R,m}^0(\bm r_m)}
\notag\\[-0.3em]
&\hspace{5em}{}
\log\phi_{\bm\theta_m}(\bm u_m)
\,d\bm u_m\,d\bm r_m\notag\\
&=
\iint
\phi_{\bm\theta_m^0}(\bm u_m)g_m^\star(\bm u_m\mid\bm r_m)
\notag\\[-0.3em]
&\hspace{5em}{}
\log\phi_{\bm\theta_m}(\bm u_m)
\,d\bm r_m\,d\bm u_m\notag\\
&=
\int
\phi_{\bm\theta_m^0}(\bm u_m)
\left[\int g_m^\star(\bm u_m\mid\bm r_m)\,d\bm r_m\right]
\notag\\[-0.3em]
&\hspace{5em}{}
\log\phi_{\bm\theta_m}(\bm u_m)\,d\bm u_m\notag\\
&=
\int \phi_{\bm\theta_m^0}(\bm u_m)
\log\phi_{\bm\theta_m}(\bm u_m)\,d\bm u_m.
\label{eq:proof-prior-q}
\end{align}
Consequently,
\begin{align}
Q_m(\bm\theta_m^0)-Q_m(\bm\theta_m)
&=
\int \phi_{\bm\theta_m^0}(\bm u_m)
\log
\frac{\phi_{\bm\theta_m^0}(\bm u_m)}
     {\phi_{\bm\theta_m}(\bm u_m)}
\,d\bm u_m\notag\\
&= 
D_{\mathrm{KL}}\!\left(
\phi_{\bm\theta_m^0}\,\middle\|\,\phi_{\bm\theta_m}
\right)\ge 0.
\label{eq:proof-prior-kl}
\end{align}

For a likelihood module, write
$\phi_{\bm\theta_m}(\bm y_m\mid\bm u_m)
=\phi_m(\bm y_m\mid\bm u_m;\bm\theta_m)$ and
$\phi_{\bm\theta_m^0}(\bm y_m\mid\bm u_m)
=\phi_m(\bm y_m\mid\bm u_m;\bm\theta_m^0)$. On the LMMSE side, standard VAMP SE
describes the incoming error in the right-singular-vector coordinates, i.e.,
$\bm V^\top(\bm r_m-\bm u_m)$ is asymptotically Gaussian at the matched fixed
point \cite[Thm.~1]{Rangan2019VAMP}. This motivates the Gaussian cavity used by
the likelihood/LMMSE module. Let $g_m^\star(\bm u_m\mid\bm r_m)$ denote this
matched Gaussian message/cavity density, written in the original coordinates as
$\mathcal N(\bm u_m;\bm r_m,v_m^\star\bm V\bm V^\top)$. Together with the matched
likelihood factor, it defines a local cavity law conditional on $\bm r_m$,
\begin{equation}
p_m^0(\bm u_m,\bm y_m\mid\bm r_m)
=
g_m^\star(\bm u_m\mid\bm r_m)\,
\phi_{\bm\theta_m^0}(\bm y_m\mid\bm u_m),
\label{eq:proof-likelihood-cond-law}
\end{equation}
under which $\bm Y_m$ and $\bm R_m$ are conditionally independent given
$\bm U_m$. Its tilted normalizer is the conditional evidence
\begin{align}
p_m^0(\bm y_m\mid\bm r_m)
&=
\int
\phi_{\bm\theta_m^0}(\bm y_m\mid\bm u_m)
g_m^\star(\bm u_m\mid\bm r_m)
\,d\bm u_m .
\label{eq:proof-likelihood-normalizer}
\end{align}
Let $p_{R,m}^0$ be the matched population law of the incoming cavity message.
Then the corresponding local law of $(\bm Y_m,\bm R_m)$ is
\begin{equation}
p_m^0(\bm y_m,\bm r_m)
=
p_{R,m}^0(\bm r_m)p_m^0(\bm y_m\mid\bm r_m).
\label{eq:proof-likelihood-yr-law}
\end{equation}
Thus
\begin{equation}
\tilde p_m^0(\bm u_m\mid\bm y_m,\bm r_m)
=
\frac{
g_m^\star(\bm u_m\mid\bm r_m)
\phi_{\bm\theta_m^0}(\bm y_m\mid\bm u_m)}
{p_m^0(\bm y_m\mid\bm r_m)}
\label{eq:proof-likelihood-tilt}
\end{equation}
is the conditional density of $\bm U_m$ given $(\bm Y_m,\bm R_m)$ under
\eqref{eq:proof-likelihood-cond-law} and \eqref{eq:proof-likelihood-yr-law}. The
likelihood M-step is analyzed under this local cavity law, not under the
original data-generating law. Hence the population M-step objective is
\begin{align}
Q_m(\bm\theta_m)
&=
\iiint
p_{R,m}^0(\bm r_m)
g_m^\star(\bm u_m\mid\bm r_m)
\phi_{\bm\theta_m^0}(\bm y_m\mid\bm u_m)
\notag\\[-0.3em]
&\hspace{2em}{}
\log\phi_{\bm\theta_m}(\bm y_m\mid\bm u_m)
\notag\\[-0.3em]
&\hspace{2em}{}
d\bm y_m\,d\bm u_m\,d\bm r_m .
\label{eq:proof-likelihood-objective}
\end{align}
Equivalently, one may start from the outer law
\eqref{eq:proof-likelihood-yr-law}; substituting \eqref{eq:proof-likelihood-tilt}
then cancels $p_m^0(\bm y_m\mid\bm r_m)$. Therefore,
\begin{align}
&Q_m(\bm\theta_m^0)-Q_m(\bm\theta_m)
\notag\\
&=
\iint
p_{R,m}^0(\bm r_m)g_m^\star(\bm u_m\mid\bm r_m)
\notag\\[-0.3em]
&\hspace{2em}{}
D_{\mathrm{KL}}\!\left(
\phi_{\bm\theta_m^0}(\cdot\mid\bm u_m)\,\middle\|\,
\phi_{\bm\theta_m}(\cdot\mid\bm u_m)
\right)
d\bm u_m\,d\bm r_m
\ge 0.
\label{eq:proof-likelihood-kl}
\end{align}
Together with prior identifiability and the likelihood identifiability
assumption under the cavity-induced weighting, equality in either case implies
\begin{equation}
\bm\theta_m=\bm\theta_m^0.
\label{eq:proof-ident}
\end{equation}
Thus the population M-step satisfies
\begin{equation}
\mathcal M_m(\bm\theta^0,\bm v^\star)=\bm\theta_m^0.
\label{eq:proof-mstep-fixed}
\end{equation}
The damping step preserves the true parameter,
\begin{equation}
(1-\beta)\bm\theta^0+\beta\bm\theta^0=\bm\theta^0,
\label{eq:proof-damping}
\end{equation}
so $(\bm\theta^0,\bm v^\star)$ is a fixed point of the population adaptive
recursion. At this point the recursion is the matched SC-VAMP/VAMP recursion;
the replica Bayes-optimality of $\hat{\bm x}$ follows from (A3)--(A4) and
\cite[Thm.~3]{Rangan2019VAMP}.
\end{proof}
Proposition~\ref{prop:fp} places the true parameters at a fixed
point but does not address its local stability, that is, whether nearby
parameter values are attracted to it; a stability analysis is left to future
work.

\section{Numerical Results}
\label{sec:results}

We compare three variants throughout. The oracle runs SC-VAMP with the true
parameters. The adaptive method starts from mismatched parameters and
updates them using the proposed EM steps. The frozen baseline uses the
same mismatched initialization as the adaptive method but disables all EM
updates, isolating the value of self-tuning.
The iteration panels also show the oracle SE prediction: the
matched state-evolution recursion run at the true parameters, plotted as
markers.
Performance is reported by the
normalized mean-squared error
\[
\mathrm{NMSE}=\|\hat{\bm x}-\bm x\|^2/\|\bm x\|^2 .
\]
The iteration curves are averaged over independent Monte-Carlo trials. In the
SNR sweeps, adaptive and frozen show the geometric mean over trials with
$\pm1\sigma$ error bars in the $\log_{10}$ domain (natural for log-scale NMSE),
and the oracle curve shows the median. In the randomized initialization
experiments, adaptive and frozen share the same initial draw in each trial.

\subsection{Linear CS}
For the linear model $\bm y=\bm A\bm x+\bm w$, $\bm A$ is row-orthogonal with
$N{=}2000$ and $M{=}1000$. The signal follows the BG prior with
$\rho{=}0.1$ and $\sigma_x^2{=}1$, and the noise variance is set from the
specified SNR. The adaptive method learns the prior parameters
$(\rho,\sigma_x^2)$ in the prior module and the noise variance $\sigma_w^2$ in
the linear observation module. Each run uses $25$ SC-VAMP iterations.

Figures~\ref{fig:lin}(a) and \ref{fig:lin}(b) use a fixed mismatched
initialization at SNR $20$~dB:
$\rho$ is initialized at $3\times$ its true value, $\sigma_x^2$ at
$0.5\times$, and $\sigma_w^2$ at $4\times$. Averaged over $50$ Monte-Carlo
trials, adaptive SC-VAMP quickly approaches the oracle NMSE, whereas the frozen
baseline converges to a visibly worse fixed point. The parameter trajectories
show why: $(\sigma_w^2,\rho,\sigma_x^2)$ all move close to their true values
within roughly $15$ iterations. The undamped update $\beta{=}1$ reaches the
truth faster, while $\beta{=}0.5$ gives smoother transients.

Figure~\ref{fig:lin}(c) tests robustness rather than a single hand-picked
initialization. For each SNR and each of $1000$ trials, we draw
$\rho\in[0.02,0.6]$ and draw both $\sigma_x^2$ and $\sigma_w^2$ within
$\pm10$~dB of their true values. Adaptive and frozen start from the same draw.
Across the full SNR range, adaptive remains essentially on the oracle curve,
while frozen mismatch is about $1.5$--$2.5\times$ worse.

\subsection{One-bit CS}
For one-bit CS, the measurements are
$\bm y=\mathrm{sign}(\bm A\bm x+\bm w)$ with $N{=}1000$ and $M{=}2000$. We use
the same BG prior with $\rho{=}0.1$. Because one-bit observations identify only
the scale ratio $\sigma_x^2/\sigma_w^2$, we fix $\sigma_w^2{=}1$ and tune
$(\rho,\sigma_x^2)$. The input SNR is
$10\log_{10}(\rho\sigma_x^2/\sigma_w^2)$, measured before the sign map. The
one-bit likelihood module uses the probit factor
$p(y_i\mid z_i;\sigma_w^2)=\Phi(y_i z_i/\sigma_w)$, and the EM update is damped
with $\beta{=}0.3$ for stability. The fixed-start experiments run $40$ SC-VAMP
iterations; the SNR sweep uses $120$.

Figures~\ref{fig:1bit}(a) and \ref{fig:1bit}(b) show the fixed-start experiment
at input SNR $10$~dB, averaged over $50$ trials. The initialization is strongly
mismatched: $\rho$ starts at $0.30$ and $\sigma_x^2$ at $0.30$ times its true
value. Even with the nonlinear quantizer, adaptive SC-VAMP tracks the oracle
NMSE, while frozen mismatch remains above it; the parameter plot shows both
$\rho$ and the identifiable scale being recovered.

Figure~\ref{fig:1bit}(c) sweeps input SNR over $1000$ randomized trials per
point. We draw $\rho\in[0.02,0.5]$ and draw $\sigma_x^2$ within approximately
$\pm3$~dB of the true value; the narrower scale range reflects the stronger
scale sensitivity of the probit likelihood. Adaptive consistently improves over
the frozen baseline, from about $1.2\times$ lower NMSE at $-10$~dB to
$3.1\times$ at $10$~dB, the frozen penalty growing with SNR, where the wrong
prior scale matters more. The sweep stops at $10$~dB: beyond it the sign map is
nearly scale-invariant, so $\sigma_x^2$ is no longer reliably identifiable from
one-bit data.

\section{Conclusion}
Adaptive SC-VAMP attaches an EM update to each parameterized factor, reusing
moments the SISO interface has already computed. It leaves the modular
mean--variance interface intact, has the true parameters as a Bayes-optimal
population fixed point (Proposition~\ref{prop:fp}), and approaches oracle
performance on linear and one-bit BG CS from strong mismatch. Synthetic BG models are used by design: the
oracle parameters and the matched SE prediction are then exactly known, which
is what makes the oracle/frozen comparison meaningful. Validation on measured
channel models, and learned-score priors, for which
self-tuning must rely on the tilted moments alone, remain
future work.

\section*{Acknowledgment}
This work was supported by JST, CRONOS, Japan Grant Number JPMJCS25N5.

\bibliographystyle{IEEEtran}
\bibliography{refs}

\end{document}